\documentclass[twocolumn,superscriptaddress,prx,showpacs]{revtex4-2}
\usepackage{graphicx}
\usepackage{epsfig}
\usepackage{float}
\usepackage{amssymb,amsmath,amsfonts,hyperref}
\usepackage{wasysym}
\usepackage{latexsym}
\usepackage{verbatim}
\usepackage[normalem]{ulem}
\usepackage{color}
\usepackage{tikz}
\usepackage{tikzscale}
\usepackage{enumitem}
\usepackage{fmtcount}
\usepackage{amsthm}

\newcommand{\be}{\begin{eqnarray}}
	\newcommand{\ee}{\end{eqnarray}}

\newcommand{\calR}{{\mathcal R}}
\newcommand{\calP}{{\mathcal P}}

\newtheorem{theorem}{Theorem}

\theoremstyle{definition}

\begin{document}
	\title{Random geometry of maximum-density dimer packings of the site-diluted kagome lattice}
	\author{Ritesh Bhola}
	\affiliation{\small{Department of Theoretical Physics, Tata Institute of Fundamental Research, Mumbai 400005, India}}
	\author{Kedar Damle}
	\affiliation{\small{Department of Theoretical Physics, Tata Institute of Fundamental Research, Mumbai 400005, India}}
	\begin{abstract}
	Recent work~\cite{Ansari_Damle_2024} that analyzed the effect of vacancy disorder on a short-range resonating valence bond spin liquid state  of kagome-lattice antiferromagnets argued that such spin liquids are stable to vacancy disorder. The argument relied crucially on a numerical study that identified the following property of the site-diluted kagome lattice: maximum-density dimer packings (maximum matchings) of any connected component of such site-diluted kagome lattices have at most one unmatched vertex that hosts a monomer. Here, we provide an inductive proof of a stronger result that implies this property: If a connected cluster of such a lattice has an odd number of vertices, its Gallai-Edmonds decomposition~\cite{Lovas_Plummer_1986} has  exactly one ${\mathcal R}$-type region that spans the entire connected cluster and hosts a single monomer of any maximum-density dimer packing. If on the other hand it has an even number of sites, it admits perfect matchings (fully-packed dimer coverings with no monomers) and its Gallai-Edmonds decomposition consists of a single ${\mathcal P}$-type region that spans the entire cluster.  Our proof also applies to the site-diluted Archimedean star lattice, the site-diluted pyrochlore lattice (corner-sharing tetrahedra), the site-diluted hyperkagome lattice, and, more generally, to any lattice satisfying a certain local connectivity property. It does not apply to bond-diluted versions of such lattices. Indeed, a simple argument can be used to obtain a {\em lower bound} on the {\em density} of unmatched vertices in maximum matchings of bond-diluted kagome lattices.
	\end{abstract}
	\maketitle
	\section{Introduction}
	In statistical mechanics, dimer models serve as an interesting and important paradigm for the behavior of macroscopic systems which have a thermodynamically-significant degeneracy of 
low-energy states~\cite{Henley_2010,Huse_Krauth_Moessner_Sondhi_2003,Alet_Ikhlef_Jacobsen_Grégoire_2006,Desai_Pujari_Damle_2021,Mambrini_2008,Rakala_Damle_Algorithm,Damle_Dhar_Ramola_2012,Patil_Dasgupta_Damle_2014} More accurately, one studies the ensemble statistics of fully-packed dimer covers of a lattice, which obey two constraints: No site is left untouched by a dimer, and no two dimers touch at a site. Here, dimers are hard rods that can be placed on links of the lattice and touch the two sites connected by the link. In computer science, such fully-packed dimer coverings are called ``perfect matchings'' and sites touched by the same dimer are said to be matched to each other~\cite{Lovas_Plummer_1986}. Algorithms for diagnosing the existence of perfect matchings of a graph and producing one such matching are a well-studied and important part of this subject~\cite{Lovas_Plummer_1986}. 
	
	In the presence of quenched disorder, the twin constraints of full packing and hard-core exclusion are typically incompatible, and the best one can do is a ``maximum matching'', equivalently, a maximum-density dimer packing~\cite{Bhola_Biswas_Islam_Damle_2022,Lovas_Plummer_1986}. Such a dimer packing has a minimum number of monomers (located on the unmatched vertices of the maximum matching). Although the statistical mechanics of maximum-density dimer packings has only recently entered the physics literature, the problem of finding maximum matchings is an important and extremely well-studied part of the computer science literature, with connections to resource allocation problems, network flow optimization, and algorithm design~\cite{Lovas_Plummer_1986,Moret_Shapiro_1991}. Further, the structure of maximum matchings plays an important role in graph theory~\cite{Lovas_Plummer_1986}.

	Recent work~\cite{Bhola_Biswas_Islam_Damle_2022} has shown interesting connections between the structure of maximum matchings of disordered bipartite lattices and several problems of interest in condensed matter physics. Using the Dulmage-Mendelsohn decomposition theorem, the authors showed that a disordered bipartite lattice can be uniquely decomposed into a complete set of non-overlapping  monomer-carrying regions, dubbed $\calR$-type regions, and perfectly matched regions, dubbed $\calP$-type regions~\cite{Bhola_Biswas_Islam_Damle_2022}. The $\calR$-type regions of a bipartite lattice can be further classified into $\calR_A$-type and $\calR_B$-type regions, since any given region can only host monomers on one of the two sublattices of the bipartite lattice~\cite{Bhola_Biswas_Islam_Damle_2022}. 	Other work~\cite{Wang_Sandvik_2006,Wang_Sandvik_2010} had also noted an interesting connection between monomer-carrying regions of maximum-density dimer packings and the spatial form factor of low-energy triplet excitations in the antiferromagnetically ordered phase of the spin $S=1/2$ Heisenberg antiferromagnet on diluted bipartite lattices. The theoretical picture in terms of ${\mathcal R}$-type regions and their random geometry~\cite{Bhola_Biswas_Islam_Damle_2022} provides a systematic basis for obtaining spatial information about these excitations.
	
	These $\calR$-regions were shown to also host the topologically-protected zero-energy eigenstates of disordered tight-binding models whose nonzero hopping amplitudes are specified by the surviving links of the bipartite lattice, and whose orbitals correspond to surviving sites of the bipartite lattice. The construction of $\calR$-type and $\calP$-type regions can also be generalized to non-bipartite graphs using the Gallai-Edmonds decomposition~\cite{Damle_2022}. Importantly, the $\calR$-type regions thus constructed host the topologically-protected collective Majorana zero modes of the corresponding Majorana network Hamiltonion~\cite{Damle_2022}. At low dilution, these regions of site-diluted non-bipartite lattices such as the triangular or Shastry-Sutherland lattice in two dimensions and the stacked triangular or corner-sharing octahedral lattices in three dimensions undergo an interesting percolation transition, and the corresponding percolated phase exhibits unusual violations of self-averaging~\cite{Bhola_Damle2025Preprint}.

In a separate strand of work, the recent study of Ref~\cite{Ansari_Damle_2024} investigated the role of vacancy disorder on quantum paramagnetic phases like short-range resonating valence bond (sRVB) spin liquid phases and valence bond solids (VBS). The key finding was that sRVB spin liquid phases develop a local-moment instability whenever the density $w$ of monomers in the corresponding maximum-density dimer packings is nonzero. These emergent local moments control the low-energy physics and the final fate of the system is determined by the interactions between them, which in turn depends crucially on the spatial form factor of these local moments. The random geometry of ${\mathcal R}$-type regions is again significant, since these emergent local moments live within the ${\mathcal R}$-type regions and interactions and the dominant interactions between them are expected to lie within individual ${\mathcal R}$-type regions.  This is in sharp contrast with the fact that VBS phases are {\em always} unstable to vacancy disorder since every vacancy seeds at intermediate energy scales a local moment in its vicinity.

Perhaps the most significant conclusions of their work was that sRVB spin liquid states of kagome lattice antiferromagnets are expected to be {\em stable} for small dilution $n_v$, in the sense an nonzero $n_v$ does not lead to a local moment instability. The key input here was a remarkable property of the site-diluted kagome lattice, established on the basis of numerical evidence in Ref.~\cite{Ansari_Damle_2024}: A maximum-density dimer packing of any geometrically-connected cluster of the site-diluted kagome lattice has at most one monomer independent of the site of the connected cluster. Indeed, connected clusters with an even number of sites were found to admit perfect matchings, while those with an odd number of sites always had maximum matchings with a single monomer. 
Given the tight connection between monomers of maximum-density dimer packings and vacancy-induced emergent local moments in sRVB spin liquid states, it would also have been useful to have additional spatialy-resolved information about where the solitary monomer is allowed to live inside odd-cardinality clusters. Indeed, such information is expected to have implications for exact-diagonalization studies of the low-energy spectrum of such clusters if the magnetic Hamiltonian or small deformations thereof support such spin liquid ground states. 
However, the numerical study was not able to obtain any such spatially resolved information.

With this motivation, we provide here an inductive proof of not just the basic observation of Ref.~\cite{Ansari_Damle_2024}, but also a characterization of the Gallai-Edmonds decomposition of connected clusters of the site-diluted kagome lattice. The latter is enabled by our proof-strategy, which uses the structure theory of Gallai and Edmonds in a crucial way apart from a certain crucial local connectivity property of the kagome lattice, which is preserved under site dilution but not by bond dilution. Although we present the proof in explicit detail only for the kagome lattice, our arguments also apply, mutatis mutandis, to any other lattice that satistifes this key local connectivity property. This class of lattices includes the star lattice (decorated honeycomb lattice), the pyrochlore lattice (comprising corner-sharing tetrahedra)  and a class of lattices in dimensions three and higher that are extensions of the two-dimensional kagome lattice~\cite{Chandra_Dhar_2008}. More generally, we also explain how this property is shared by any site-diluted lattice whose parent is constructed as the line graph of another lattice. In this language, the kagome lattice is the line graph of the honeycomb net, the pyrochlore lattice is the line graph of the diamond lattice, and the star lattice is the line graph of the Lieb lattice constructed by adding two vertices to every link of the honeycomb net.

The rest of the paper is organized as follows: In Sec.~\ref{sec:GE}, we provide a quick summary of the Gallai-Edmonds decomposition theorem and associated structure theory, as well as the construction of $\calR$-type and $\calP$-type regions of a general graph. In Sec.~\ref{sec:proof}, we present our inductive proof, explain how it applies to a fairly large class of site-diluted lattices, and set our results in the context of previous work in the mathematical literature that anticipated one part of the theorem proved here. In addition, in Sec.~\ref{sec:blossoms}, we provide additional numerical results that help us understand, as a function of the dilution $n_v$,  the internal structure of the single ${\mathcal R}$-type region that spans the entirety of the largest geometrically-connected cluster of a site-diluted kagome lattice.  


\section{Gallai-Edmonds decomposition: $\calR$-type and $\calP$-type regions }
\label{sec:GE}
This section summarizes key elements of the structure theory of Gallai and Edmonds~\cite{Lovas_Plummer_1986}, and the construction of $\calR$-type and $\calP$-type regions~\cite{Bhola_Biswas_Islam_Damle_2022,Damle_2022,Bhola_Damle2025Preprint} based on the Gallai-Edmonds decomposition on a general graph. This decomposition theorem partitions the vertices of a graph into three unique sets: even-type sites ($E$), odd-type sites ($O$) and unreachable sites ($U$) starting from any one of the maximum matchings of the graph. Even-type vertices are those that are reachable from a monomer by an even-length alternating path (aternating between links unoccupied by a dimer of a maximum matching and occupied by such a dimer). Odd-type vertices are reachable by such alternating paths in an odd number of steps, but not an even number of steps.
And unreachable vertices, as the name suggests, are simply not reachable from a monomer of a maximum matching by such alternating paths. Gallai-Edmonds theory guarantees that these type labels are properties of the underlying graph, and independent of the particular maximum matching one uses to assign these labels.

Now, to construct $\calR$-type and $\calP$-type regions~\cite{Bhola_Biswas_Islam_Damle_2022,Damle_2022,Bhola_Damle2025Preprint}, remove all links that connect any two odd-type sites as well as all links that connect any odd-type site to an unreachable site. The resulting connected components either comprise only unreachable sites, or are made up of odd-type and even-type sites. The former are the ${\mathcal P}$-type regions that are guaranteed to be perfectly matched by any maximum matching. The latter are the monomer-carrying ${\mathcal R}$-type regions which host the monomers of any maximum matching. These monomers always live on even-type sites. These even-type sites are organized in odd-cardinality structures called blossoms, and any given odd-type site is guaranteed to be connected to more than one blossom. Moreover, any maximum matching always matches every odd-type site to some even-type site, and no two odd-type sites are matched to even-type sites of the same blossom. In any maximum matching, each blossom is either perfectly matched (if one of its even-type sites is matched to an odd-type site) or hosts exactly one monomer (if none of its sites are matched to any odd-type site).
Given any even-type site, it is always possible to find a maximum matching that places a monomer on that even-type site. Also, in any maximum matching, all sites of a given ${\mathcal R}$-type region are reachable via alternating paths from one of the monomers contained in it.

	\section{Theorem on the Gallai-Edmonds decomposition of a class of site-diluted lattices}
	\label{sec:proof}
We provide here an inductive proof of a theorem on the structure of the Gallai-Edmonds decomposition of a class of site-diluted lattices. The statement was formulated based on the numerical results of Ref.~\cite{Ansari_Damle_2024} on the number of monomers in maximum-density dimer packings of connected clusters of a site-diluted kagome lattice.  To connect with these numerical results, we first formulate and prove the result for the site-diluted kagome lattice in Sec.~\ref{subsec:kagome}, and then indicate in Sec.~\ref{subsec:generalizations} how it generalizes to a large class of lattices of interest in physics, which all share a certain local connectivity property of the kagome lattice that is preserved under site dilution. 
After this work was completed, and while it was being presented in a seminar for the first time, Piyush Srivastava made us aware of the fact that this local connectivity property we had identified has a name and a body of theory associated with it in the mathematical literature: Graphs which satisfy this property are called ``claw-free'' graphs. Armed with this additional knowledge (and access to bibliographic search tools), we conclude this section by setting our result in the proper mathematical context~\cite{FaudreeClawFreeReview1997} and explaining how it goes beyond a classic theorem~\cite{SumnerAMS1974,LasVergnas1975} proved half a century ago in the graph theory literature. 
   \begin{figure}[ht]
   	\includegraphics[width=\columnwidth]{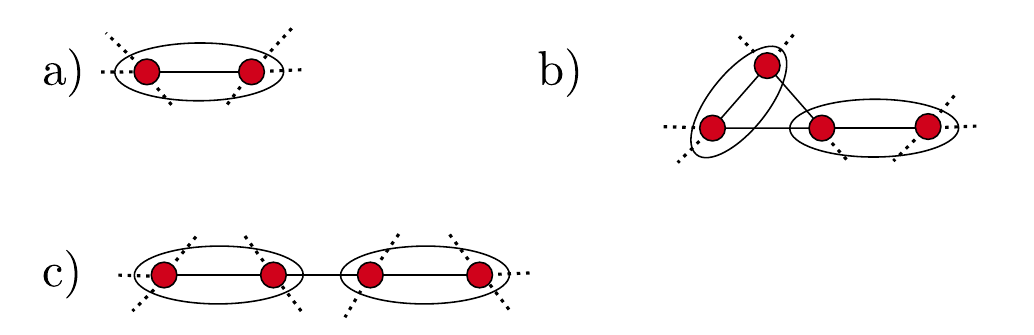}
   	\caption{The above figure shows all possible even parity clusters of size two and four on the kagome lattice. Notice that a single site, when attached to these regions, makes them a $\calR$-type region.}
   	\label{Fig:P_regions}
   \end{figure}
   \begin{figure}[ht]
   	\includegraphics[width=\columnwidth]{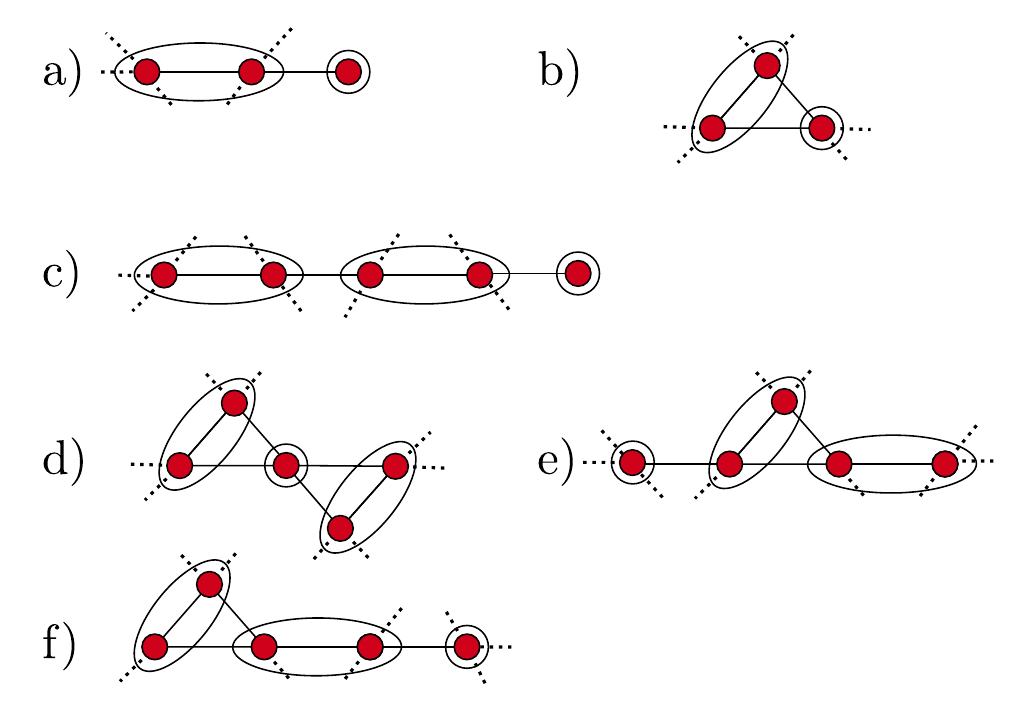}
   	\caption{The above figure shows all possible odd parity clusters of size three and five on the kagome lattice. Notice that these motifs are $\calR$-type regions with one monomer. a) and b) can be generated by attaching a single site to a) of Fig.~\ref{Fig:P_regions} and similarly c), d), e) and f) can be generated by attaching a single site to b) and c) of Fig.~\ref{Fig:P_regions}.  }
   	\label{Fig:R_regions}
   \end{figure}

\subsection{On the site-diluted kagome lattice}	
\label{subsec:kagome}
	\begin{theorem}
	 For $n\geq 1$, every odd-parity connected cluster $C_{2n-1}$ comprising $2n-1$ vertices of the site-diluted kagome lattice admits a near-perfect maximum matching with exactly one unmatched vertex that hosts a lone monomer. Moreover, its Gallai-Edmonds decomposition consists of a single $\calR$-type region. For $n\geq 2$, every even-parity connected cluster $C_{2n-2}$ comprising $2n-2$ vertices of the site-diluted kagome lattice admits a perfect matching; its Gallai-Edmonds decomposition thus has a single $\calP$-type region. 
	\end{theorem}	
	
	\begin{proof} 
		By induction:\\
		\begin{itemize}
			
			\item \textbf{Establishing the base case:} For $n=1$, the only odd-parity cluster $C_{1}$ is an isolated site,
			 which is of course trivially an ${\mathcal R}$-type region that hosts a single monomer.
			  For $n=2$, an odd-parity cluster $C_3$ of the site-diluted kagome lattice can have the two possible inequivalent (in terms of connectivity) configurations shown in Figs.~\ref{Fig:R_regions}(a) and (b). By inspection of the particular maximum-density dimer cover displayed in the figure, we see immediately that the entire cluster is a single $\calR$-type region that hosts a single monomer of any maximum-density dimer cover. Finally, for $n=3$, there are four inequivalent (in terms of connectivity) odd-parity clusters $C_5$ possible on the site-diluted kagome lattice. These are shown in Fig.~\ref{Fig:R_regions}(c)---(f). By direct inspection of one of the maximum-density dimer covers of each of these clusters, we see that each of them comprises a single $\calR$-type region that spans the whole cluster. Thus, the base case is established by direct inspection for $n=1, 2, 3$ for odd-parity clusters. Turning to even-parity clusters, the only even-parity cluster $C_{2}$ consists of two neighboring sites
			   (Fig.~\ref{Fig:P_regions}(a)), which can obviously be matched with each other, and $C_2$ is thus
			    trivially a ${\mathcal P}$-type region.  Likewise, there are two possible inequivalent (in connectivity terms) even-parity clusters $C_4$, shown in Fig.~\ref{Fig:P_regions}(b) and (c). These clusters also admit a perfect matching as shown in the figure, and therefore are a ${\mathcal P}$-type region. Thus, the statement of the theorem is trivially valid for $n=2$ and $n=3$, which constitute the base cases for the even-parity clusters.
			
			\item \textbf{Induction hypothesis:} The statement of the theorem is true for all even-parity clusters of size less than or equal to $2n-2$ and all odd-parity clusters of size less than or equal to $2n-1$. Since we have checked the base cases, we know this hypothesis is true for $n=3$.
			
			\item \textbf{Induction:} \\  
 We first ask for the distinct ways in which one obtains an even-parity cluster $C_{2n}$ of size $2n$ from an odd-parity cluster of size $C_{2n-1}$ which is a single ${\mathcal R}$-type region. The key constraint here is that we are working with a site-diluted kagome lattice, so if there are two adjacent sites that survive the dilution, they are necessarily connected by a nearest-neighbor bond. This constraint brings into play the following interesting property of the kagome lattice: If a site has three neighbors, {\em i.e.}, connected by a bond to three adjacent sites, then at least one pair of these neighbors is also connected to each other by a bond. This is true for any set of three neighbors of any one site. Clearly, site-dilution preserves this property, but bond-dilution does not. We will make good use of this local connectivity property in the arguments that follow.
 
Now, the site $s$ that is added to $C_{2n-1}$ has to be adjacent to at least one site of $C_{2n-1}$ in order for $C_{2n}$ to be a connected cluster. Let us denote by $b_1$ this site belonging to $C_{2n-1}$, whose existence is thus guaranteed. It can either be an odd-type site or an even-type site. We first analyze the case when it is an even-type site. In this case, Gallai-Edmonds theory guarantees us that there is a maximum matching of $C_{2n-1}$ that leaves $b_1$ unmatched, so that it hosts a monomer. We may thus start with this matching, and obtain a perfect matching of $C_{2n}$ simply by placing an additional dimer connecting  $s$ to $b_1$. Thus, in this case, $C_{2n}$ is guaranteed to admit a perfect matching, which is what we set out to prove.

Next, consider the possibilities that arise if $b_1$ is an odd-type site of $C_{2n-1}$. The Gallai-Edmonds theorem guarantees that no odd-type site has only one even-type neighbor. For if this were the case, this odd-type site would always be matched to this even-type neighbor in any maximum matching, contradicting the fact that there must exist a maximum matching in which this even-type neighbor is left unmatched.   Thus, $b_1$ must have at least two even-type neighbors. Since it is always possible to find a maximum matching that places a monomer on any even-type site, we choose a maximum matching of $C_{2n-1}$ that places a monomer on one of these even-type neighbors, which we label $b_3$. Since $b_1$ is odd-type, it must be matched by a dimer to some other even-type neighbor in this maximum matching. We label this even-type neighbor $b_2$.  
Thus, $b_1$ must have at least three neighbors: the site $s$ that belongs to $C_{2n}$ but not $C_{2n-1}$, and the two even-type sites $b_2$ and $b_3$ of $C_{2n-1}$, of which $b_3$ hosts a monomer of our chosen maximum matching and $b_2$ is matched to $b_1$ by this maximum matching. 

From the local connectivity property advertised at the outset, at least two of these three sites must be adjacent to each other. Without further analysis, we thus have five possibilities for the local geometry. These are shown in Fig.~\ref{Fig:C2n-1_to_C2n}. 
If we now examine these cases, we see that Fig.~\ref{Fig:C2n-1_to_C2n}(a) is internally inconsistent, since by the very definition of even-type and odd-type sites, we see that $b_1$ would in fact be an even-type site of $C_{2n-1}$ in this case, contradicting our assumption that it is an odd-type site of this odd-parity cluster. 
The other cases shown in Fig.~\ref{Fig:C2n-1_to_C2n}(b)---(e) are internally consistent. And in each of these cases, we see that there is an alternating path that goes from the monomer on the even-type site $b_3$ to the new site $s$. This means we can use this alternating path to add one dimer to the chosen maximum matching of $C_{2n-1}$ to obtain a perfect matching of $C_{2n}$.
This now establishes the claim when $b_1$ is an odd-type site, completing the proof of the first part of the theorem.
			\begin{figure}[ht]
				\centering
				\includegraphics[width=0.8\columnwidth]{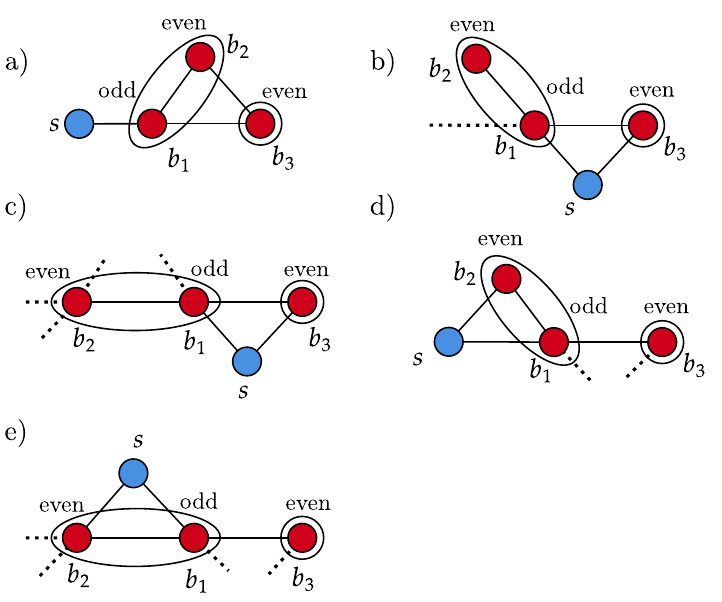}
				\caption[All possible configurations when a single site $s$ is attached to a two-coordinated odd site $b_1$ of a $\calR$-type region on the kagome lattice.]{All possible configurations are shown when a single site $s$ is attached to a two-coordinated odd site $b_1$ of a $\calR$-type region on the kagome lattice. $b_2$ and $b_3$ are neighbours of $b_1$ with $b_2$ being the matched neighbour of $b_1$. Notice that configurations b) and c) are similar, and configurations d) and e) are also similar.}
				\label{Fig:C2n-1_to_C2n}
			\end{figure}	
						\begin{figure}[ht]
				\centering
				\includegraphics[width=0.8\columnwidth]{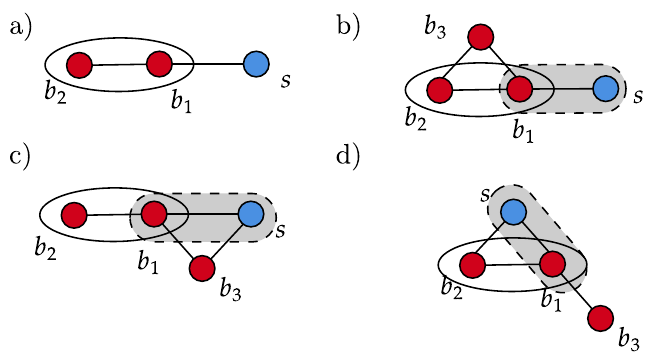}
				\caption[All possible configurations are shown when a single site $s$ is removed to one-coordinated and two-coordinated site $b_1$ of a $\calP$-type region on the kagome lattice.]{(a) A single site $s$ is attached to a site $b_1$, with only one neighbour $b_2$, of a $\calP$-type region on the kagome lattice. (b,c,d) All unique configurations are shown when a single site $s$ is attached to a two-coordinated site $b_1$ of a $\calP$-type region on the kagome lattice. $b_2$ and $b_3$ are neighbours of $b_1$ with $b_2$ being the matched neighbour of $b_1$.}
				\label{Fig:C2n_to_C2n+1}
			\end{figure}

			To prove the second part of the theorem, we need to argue that adding a site $s$ to even-parity cluster $C_{2n}$ that admits (by the first part of the theorem) a perfect matching, always gives an odd-parity cluster $C_{2n+1}$ whose maximum matchings have a single unmatched site and which consists of a single ${\mathcal R}$-type region that spans all of $C_{2n+1}$. To prove this assertion, we will use in a crucial way the additional fact that any smaller odd-parity cluster has already been shown to be spanned by a single ${\mathcal R}$-type region that hosts a single monomer, and any smaller even-parity cluster is known to host a perfect matching.
			
			Let the site of $C_{2n}$ that is adjacent to $s$ be $b_1$. Let $b_2$ be the neighbor of $b_1$ to which it is matched in $C_{2n}$ by a perfect matching of $C_{2n}$. If $b_2$ is the only neighbor of $b_1$ in $C_{2n}$ (as shown in Fig.~\ref{Fig:C2n_to_C2n+1}(a)), then $b_2$ is necessarily an even-type site of the cluster $C_{2n-1}$ obtained by deleting $b_1$ from $C_{2n}$. This is because the perfect matching of $C_{2n}$ that we started with induces a matching of $C_{2n-1}$ with exactly one monomer which is placed on $b_2$. Since $C_{2n-1}$ is a single ${\mathcal R}$-type region that hosts exactly one monomer of any maximum matching, this induced matching is a maximum matching, and $b_2$ being unmatched by it, must be an even-type site of $C_{2n-1}$. Now add both $b_1$ and $s$ to this cluster $C_{2n-1}$ and start with the maximum matching of $C_{2n-1}$ that leaves $b_2$ unmatched. We can now construct a maximum matching of $C_{2n+1}$ by adding a dimer that matches $b_1$ and $s$. In this maximum matching, there is an alternating path from the monomer on $b_2$ that goes to $s$ via $b_1$. This, in conjunction with the previously established property of $C_{2n-1}$ proves that $C_{2n+1}$ is completely spanned by a single ${\mathcal R}$-type region that hosts exactly one monomer. This proves what we set out to establish when $b_2$ is the only neighbor of $b_1$ in $C_{2n}$.
			
Next consider the possibility that $b_1$ has one other neighbor in $C_{2n}$ apart from $b_2$.
Denote this neighbor by $b_3$. In this case, $s$, $b_2$ and $b_3$ are all neighbors of $b_1$. By the previously advertised local connectivity property of the site-diluted kagome lattice, at least two of these must be connected to each other by a nearest-neibhor bond. The various inequivalent ways in which this can happen are shown in Fig.~\ref{Fig:C2n_to_C2n+1}(b)-(d). Again, examine what happens when we delete $b_1$ from $C_{2n}$. In this case, $C_{2n}$ can be reduced to a single smaller odd-parity cluster $C_{2n-1}$, as in Fig.~\ref{Fig:C2n_to_C2n+1}(b), or break up into two smaller clusters as shown in Fig.~\ref{Fig:C2n_to_C2n+1}(c), (d)). In the former case, our previous argument goes through.

In the latter cases, one of these is necessarily an even-parity cluster that can be perfectly matched (by the previously established property of smaller even-parity clusters) and the other is necessarily an odd-parity cluster that is spanned by a single ${\mathcal R}$-type region that hosts exactly one monomer of any maximum matching (again, by the previously established property of smaller odd-parity clusters). 
Since a perfect matching of $C_{2n}$ placed a dimer connecting $b_1$ to $b_2$, the cluster containing $b_3$ is perfectly matched by this matching, and hence an even-parity cluster. Therefore the cluster containing $b_2$ must be odd-parity. We can now use apply the previous argument to this cluster, which admits a maximum matching with exactly one monomer that can be located at $b_2$. Reinstating $b_1$ and $s$ and the perfectly-matched even-parity cluster containing $b_3$, we then obtain a maximum matching of $C_{2n+1}$ that has an alternating path from the monomer on $b_2$ to the new site $s$. This gives us a maximum matching of $C_{2n+1}$ with exactly one monomer, which can be moved to $s$. 

Now, consider cases Fig.~\ref{Fig:C2n_to_C2n+1}(c) and (d) separately. In the case shown in Fig.~\ref{Fig:C2n_to_C2n+1}(c), since $s$ is directly connected to $b_3$ by a bond, the cluster obtained by adding $s$ to the perfectly-matched cluster containing $b_3$ is also a smaller odd-parity cluster that must, by our previously established property of such clusters, be spanned by a single ${\mathcal R}$-type region that hosts exactly one monomer. In other words, starting from site $s$, the entire odd-parity cluster $C_{2n+1}$ can be reached by alternating paths, completing our proof in this case. 

In the case shown in Fig.~\ref{Fig:C2n_to_C2n+1}(d), the cluster containing $b_2$ is a smaller odd-parity cluster that hosts a maximum matching with exactly one monomer, located at $b_2$. For $C_{2n+1}$, this immediately implies the existence of a maximum matching with exactly one monomer, since the cluster containing $b_3$ (obtained by disconnecting it from the rest of $C_{2n}$ by deletion of $b_1$) admits a perfect matching, being a smaller even-parity cluster. Moreover, this monomer can be placed at $b_1$ simply by matching $b_2$ to $s$. Consider the cluster obtained by just adding $b_1$ to the perfectly matched cluster containing $b_3$. This is again a smaller odd-parity cluster, so alternating paths starting from the monomer at $b_1$ can reach all parts of this cluster containing $b_3$.
This shows that all of $C_{2n+1}$ can be reached by alternating paths starting from $s$ which can be chosen to host a monomer of a maximum matching of $C_{2n+1}$. This completes the proof in this case too.
And finally, we note that this argument immediately generalizes to cases in which $b_1$ has more than two neighbors in $C_{2n}$.

		\end{itemize}
	\end{proof}

    \subsection{Generalizations and connection to earlier mathematical results}
    \label{subsec:generalizations}
    We have presented a proof of the basic result in detail for the site-diluted kagome lattice. As is clear from this proof, the proof relies crucially on elements of Gallai-Edmonds theory. As we have also emphasized at the very outset, it also relies equally crucially on a specific local connectivity property of the site-diluted kagome lattice, that may be summarized thus:  If one considers any three surviving neighbors of a site, then at least two of these neighbors are also connected to each other by a link of the cluster. 

It is therefore interesting to ask what other lattices we encounter in our study of frustrated magnetism share these two properties. A quick check reveals that the star lattice (equivalently, the decorated honeycomb lattice) also obeys this property, as does the three-dimensional pyrochlore lattice as well as three-dimensional generalizations of the kagome lattice, which have been studied earlier in statistical mechanics. More generally, it is easy to see that these two properties are shared by any lattice which is obtained by placing sites on the links of an underlying lattice and defining two such sites to be nearest neighbors if the corresponding links emanate from a common site of this underlying lattice. Such lattices are called line graphs in the graph theory literature. Thus, for instance, the kagome lattice is the line graph of the honeycomb net, the star lattice is the line graph of the Lieb lattice obtained from the honeycomb lattice by adding two vertices to every link of the honeycomb lattice, and the pyrochlore lattice is the line graph of the three-dimensional diamond lattice. 
			\begin{figure}[ht]
				\centering
				\includegraphics[width=0.8\columnwidth]{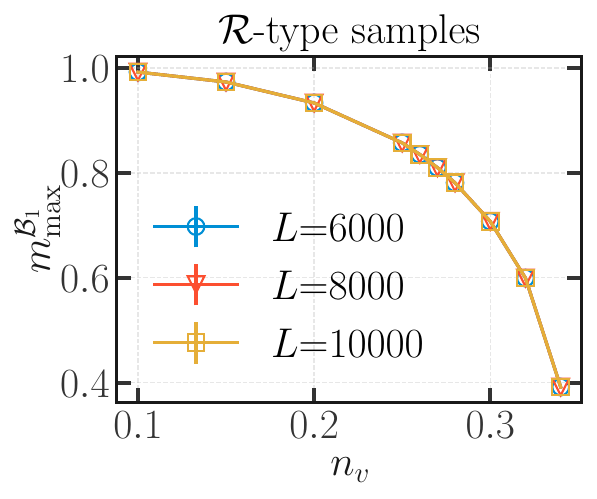}
				\caption{The largest blossom in the ${\mathcal R}$-type region that spans the largest connected cluster of a site-diluted kagome lattice occupies a nonzero fraction of the cluster, which grows rapidly as the dilution $n_v$ is reduced}
				\label{Fig:blossomsize}
			\end{figure}	
						\begin{figure}[ht]
				\centering
				\includegraphics[width=0.8\columnwidth]{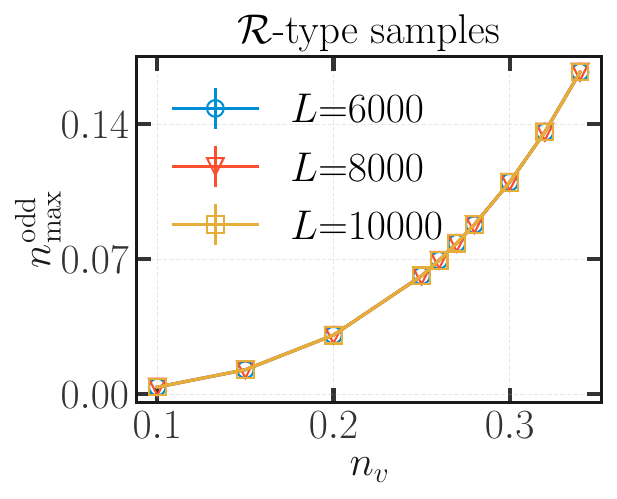}
				\caption{The number density of odd-type sites in the ${\mathcal R}$-type region that spans the largest connected cluster of a site-diluted kagome lattice goes to zero rapidly as the dilution $n_v$ is reduced. }
				\label{Fig:oddsitedensity}
			\end{figure}

As mentioned earlier, it turns out that this key local connectivity property we have exploited has a name in the mathematical literature, and a theory built around it. Graphs that satisfy it are called ``claw-free graphs''~\cite{FaudreeClawFreeReview1997}. Indeed, one part of our result was already known in this literature on claw-free graphs: Sumner~\cite{SumnerAMS1974} and Las Vergnas~\cite{LasVergnas1975} independently showed that claw-free graphs with an even number of sites necessarily admit perfect matchings. 
Now, it is of course possible to use this classic result of Sumner and Las Vergnas on even-parity clusters to show that odd-parity connected clusters of such claw-free graphs must have exactly one unmatched vertex in any maximum matching. 

However, nothing can be deduced about the Gallai-Edmonds decomposition of an odd-parity connected cluster from these classic results if one follows this route.
Our proof strategy, which uses the structure theory of Gallai and Edmonds in a crucial way, is very different and we are therefore able to obtain crucial additional information: The Gallai-Edmonds decomposition of claw free graphs with an odd number of vertices consists of a single ${\mathcal R}$-type region and no ${\mathcal P}$-type region. From earlier work, this guarantee, that the entire connected cluster is a single ${\mathcal R}$-type region, is expected to have immediate consequences for a class of localization problems on such diluted lattices, as well as for the local susceptibility of a gapped resonating valence bond spin liquid state on such diluted lattices.

\section{Structure of cluster-spanning ${\mathcal R}$-type region}
\label{sec:blossoms}
While our approach does allow us to prove that any odd-parity connected component of the site-diluted kagome lattice is completely spanned by a single ${\mathcal R}$-type region, we have not been able to prove anything about the structure of this cluster-spanning ${\mathcal R}$-type region. In particular, we are unable to obtain to obtain any rigorous results about the fraction of the cluster that is contained in the largest blossom, or the number of blossoms, odd-type sites and even-type sites, and how any of these scale with the size of the connected cluster. To partially remedy this, we have resorted to a large-scale computational study. 

We study the $L \times L$ samples of the site-diluted kagome lattice, with uncorrelated dilution $n_v$ (and $3(1-n_v)L^2$ sites on average), focusing on the Gallai-Edmonds decomposition of the largest connected component of each sample. At low $n_v$, this connected component has odd parity half the time (in the limit of large $L$). Consistent with the result we have just proved, we find that its Gallai-Edmonds decomposition always consists only of odd-type and even-type sites that form a single ${\mathcal R}$-type region. We have kept track of the mass density (measured by the number of even-type sites belonging to it divided by the total number of sites in the connected cluster) of the largest blossom in this ${\mathcal R}$-type region, as well as the number density of odd-type sites (measured as a fraction of the total number of sites in the connected cluster being studied). Together, these measurements give us a reasonably complete picture of the morphology of this ${\mathcal R}$-type region that spans the entirety of the largest connected component of the lattice at various values of dilution $n_v$. The results are shown in Fig.~\ref{Fig:blossomsize} and Fig.~\ref{Fig:oddsitedensity}. 

From this data, it is clear at small dilution $n_v$ that the largest blossom essentiallly takes over the entire ${\mathcal R}$-type region that spans the largest connected cluster of a site-diluted kagome lattice whenever this cluster has odd parity. In this low dilution regime, the number density of odd-type sites in this ${\mathcal R}$-type region shrinks rapidly to zero. It would be interesting to explore consequences of this random geometry for the physics of the lone magnetic moment that is expected to live inside such an ${\mathcal R}$-type region in the sRVB spin liquid phase of a kagome antiferromagnet, and we hope to return to this in future work.

   \section{Acknowledgments}
   We thank Piyush Srivastava for pointing us to the mathematical literature on claw-free graphs. We gratefully acknowledge generous allocation of computing resources by the Department of Theoretical Physics (DTP) of the Tata Institute of Fundamental Research (TIFR), and related technical assistance from K. Ghadiali and A. Salve. The work of RB was supported at the TIFR by a graduate fellowship from DAE, India and formed a part of his Ph.D thesis.  KD was supported at the TIFR by DAE, India, and in part by a J.C. Bose Fellowship (JCB/2020/000047) of SERB, DST India, and by
	the Infosys-Chandrasekharan Random Geometry Center
	(TIFR).

   \bibliography{references}

\end{document}